\documentclass[conference,10pt]{IEEEtran}
\usepackage{graphicx,epsfig,graphics,epsf}
\usepackage{float,amsmath,amsfonts,amssymb,amsthm, amssymb,stmaryrd,amsmath,amsfonts,rotating,mathrsfs}
\usepackage{url}
\usepackage[center,small]{caption}
\usepackage{algorithm}
\usepackage{algorithmic}

\newtheorem{thm}{Theorem}
\newtheorem{lem}{Lemma}

\begin{document}


\title{On the Construction of Polar Codes}
\author{
\IEEEauthorblockN{Ramtin Pedarsani}
\IEEEauthorblockA{School of Computer and \\ Communication Systems,\\ EPFL\\
Lausanne, Switzerland.\\
ramtin.pedarsani@epfl.ch}

\and
\IEEEauthorblockN{S. Hamed Hassani}
\IEEEauthorblockA{School of Computer and \\ Communication Systems,\\ EPFL\\
Lausanne, Switzerland.\\
seyedhamed.hassani@epfl.ch}

\and
\IEEEauthorblockN{Ido Tal}
\IEEEauthorblockA{Information Theory and\\ Applications,\\ UCSD\\
La Jolla, CA, USA.\\
idotal@ieee.org}
\and
\IEEEauthorblockN{Emre Telatar}
\IEEEauthorblockA{School of Computer and \\ Communication Systems,\\ EPFL\\
Lausanne, Switzerland.\\
emre.telatar@epfl.ch}
}
\maketitle
\begin{abstract}
We consider the problem of efficiently constructing polar codes over binary memoryless symmetric (BMS)
channels.  The complexity of designing polar codes via an exact
evaluation of the polarized channels to find which ones are ``good"
appears to be exponential in the block length.  In \cite{TV11}, Tal and Vardy
show that if instead the evaluation if performed approximately, the
construction has only linear complexity.  In this paper, we follow this
approach and present a framework where the algorithms of \cite{TV11} and new
related algorithms can be analyzed for complexity and accuracy.  We
provide numerical and analytical results on the efficiency of such
algorithms, in particular we show that one can find all the ``good"
channels (except a vanishing fraction) with almost linear complexity in
block-length (except a polylogarithmic factor). 
\end{abstract}

\section{Introduction}\label{sec:def}
\subsection{Polar Codes}
Polar coding, introduced by Ar{\i}kan in \cite{ari09}, is an encoding/decoding scheme that provably achieves the capacity of the class of BMS channels.  Let $W$ be a BMS channel. Given the rate $R<I(W)$, polar coding  is based on choosing a set of $2^nR$ rows of the matrix $G_n= \bigl [ \begin{smallmatrix} 1 & 0 \\ 1 &1  \end{smallmatrix} \bigr]^{\otimes n}$ to form a $2^nR \times 2^n$ matrix which is used as the generator matrix in the encoding procedure\footnote{There are extensions of polar codes given in \cite{SK09} which use different kinds of matrices.}.  The way this set is chosen is dependent on the channel $W$ and uses a phenomenon called channel polarization:  Consider an infinite binary tree and place the underlying  channel 
$W$ on the root node and continue recursively as follows. Having the channel $P: \{0,1\} \to \cal Y$ on a node of the tree, define the channels $P^-: \{0,1\} \to {\cal Y}^2$ and $P^+: \{0,1\} \to \{0,1\} \times  {\cal Y}^2$ 
\begin{align} \label{1}
& P^-(y_1,y_2 | x_1)= \sum_{x_2 \in \{0,1\} } \frac 12 P(y_1| x_1 \oplus x_2) P(y_2|x_2)  \\ \label{2}
& P^+(y_1,y_2,x_1 | x_2)= \frac 12 P(y_1 | x_1 \oplus x_2) P(y_2 | x_2),
\end{align}
and place $P^-$ and $P^+$ as the left and right children of this node. As a result, at level $n$ there are $N=2^n$ channels which we denote from left to right by $W_{N}^{1} $ to $W_{N}^{N}$ . In \cite{ari09}, Ar{\i}kan proved that as $n \to \infty$, a fraction approaching $I(W)$ of the channels at level $n$ have capacity close to $1$ (call them ``noiseless'' channels) and a fraction approaching $1-I(W)$ have capacity close to $0$ (call them ``completely noisy'' channels).   
Given the rate $R$, the indices of the matrix $G_n$ are chosen as follows: choose a subset of the channels $\{W_{N}^{(i)}\}_{1 \leq i \leq N}$ with the most mutual information and choose the rows $G_n$ with the same indices as these channels. For example, if the channel $W_{N}^{(j)}$ is chosen, then the $j$-th row of $G_n$ is selected, up to the bit-reversal permutation. In the following, given $n$, we call the set of indices of $NR$ channels with the most mutual information, the set of good indices.  

We can equivalently say that as $n \to \infty$ the fraction of channels with Bhattacharyya constant near $0$ approaches $I(W)$ and the fraction of  channels with Bhattacharyya constant near $1$ approaches $1-I(W)$. The Bhattacharyya constant of a channel $P: \{0,1\} \to \cal Y$ is given by
\begin{equation}
Z(P)= \sum_{y \in \cal Y} \sqrt{P(y | 0) P(y | 1)}.
\end{equation}
Therefore, we can alternatively call the set of indices of $NR$ channels with least Bhattacharyya parameters, the set of good indices. It is also interesting to mention that the sum of the Bhattacharyya parameters of the chosen channels is an upper bound on the block error probability of polar codes when we use the successive cancellation decoder.

\subsection{Problem Formulation}
\label{sec:prob}
 
Designing a polar code is equivalent to finding the set of good indices. The main difficulty in this task is that, since the output alphabet of $W_{N}^{(i)}$ is $\mathcal{Y} ^ {N} \times \{0,1\} ^i$, the cardinality of the output alphabet of the  channels at the level $n$ of the binary tree is doubly exponential in $n$ or is exponential in the block-length. So computing the exact transition probabilities of these channels seems to be intractable and hence we need some efficient methods to ``approximate'' these channels.

In \cite{ari09}, it is suggested to use a Monte-Carlo method for estimating the Bhattacharyya parameters. Another method in this regard is by \emph{quantization} \cite{TV11,Compound,Tanaka}, \cite[Appendix B]{Urbanke}: approximating the given  channel with a channel that has fewer output symbols. More precisely, given a number $k$, the task is to come up with efficient methods to replace channels that have more that $k$ outputs with ``close'' channels that have at most $k$ outputs. Few comments in this regard are the following:
\begin{itemize}
\item The term ``close'' above depends on the definition of the quantization error which can be different depending on the context. In our problem, in its most general setting we can define the quantization error as the difference between the true set of good indices and the approximate set of good indices. However, it seems that analyzing this type of error may be difficult and in the sequel we consider types of errors that are easier to analyze.  
\item Thus, as a compromise, will intuitively think of two channels as being close if they are close with respect to some given metric; typically mutual information but sometimes probability of error. More so, we require that this closeness is in the right direction: the approximated channel must be a ``pessimistic'' version of the true channel. Thus, the approximated set of good channels will be a subset of the true set.
\item Intuitively, we expect that as $k$ increases the overall error due to quantization decreases; the main art in designing the quantization methods is to have a small error while using relatively small values of $k$. However, for any quantization algorithm an important property is that as $k$ grows large, the approximate set of good indices using the quantization algorithm with $k$ fixed  approaches the true set of good indices. We give a precise mathematical definition in the sequel.    
\end{itemize}

Taking the above mentioned factors into account, a suitable formulation of  the quantization problem is to find procedures to replace each channel $P$ at each level of the binary tree with another symmetric channel $\tilde{P}$ with the number of output symbols limited to $k$ such that firstly, the set of good indices obtained with this procedure is a subset of the true good indices obtained from the channel polarization i.e. channel $\tilde{P}$ is \emph{polar degraded} with respect to $P$, and secondly the ratio of these good indices is maximized. More precisely, we start from channel $W$ at the root node of the binary tree, quantize it to $\tilde{W}$ and obtain $\tilde{W}^{-}$ and $\tilde{W}^{+}$ according to \eqref{1} and \eqref{2}. Then, we quantize the two new channels and continue the procedure to complete the tree. To state things mathematically, let $Q_k$ be a quantization procedure that assigns to each channel $P$ a binary symmetric channel $ \tilde{P}$ such that the output alphabet of $\tilde{P}$ is limited to a constant $k$. We call $Q_k$  admissible if for any $i$ and $n$
\begin{equation}\label{ad1}
I(\tilde{W}_{N}^{(i)}) \leq I(W_{N}^{(i)}).
\end{equation}
One can alternatively call $Q_k$  admissible if for any $i$ and $n$
\begin{equation}\label{ad2}
Z(\tilde{W}_{N}^{(i)}) \geq Z(W_{N}^{(i)}).
\end{equation}
Note that \eqref{ad1} and \eqref{ad2} are essentially equivalent as $N$ grows large.
Given an admissible procedure $Q_k$ and a BMS channel $W$, let $\rho(Q_k,W)$ be\footnote{Instead of $\frac 12$ in \eqref{ro} we can use any number in $(0,1)$.} 
\begin{equation} \label{ro}
\rho(Q_k,W)= \lim_{n \to \infty} \frac{| \{ i: I(\tilde{W}_{N}^{(i)}) > \frac 12 \} |}{N} 
\end{equation}
So the quantization problem is that given a number $k \in \mathbb{N}$ and a channel $W$, how can we find admissible procedures $Q_k$ such that $\rho(Q_k,W)$ is maximized and is close to the capacity of $W$. Can we reach the capacity of $W$ as $k$ goes to infinity? Are such schemes universal in the sense that they work well for all the BMS  channels? It is worth mentioning that if we first let $k$ tend to infinity and then $n$ to infinity then the limit is indeed the capacity, but we are addressing a different question here, namely we first let $n$ tend to infinity and then $k$ (or perhaps couple $k$ to $n$). In Section \ref{sec:ex}, we indeed prove that such schemes exist.  

\section{Algorithms for Quantization}\label{sec:algo}

\subsection{Preliminaries}\label{sec:pre}
Any discrete BMS channel can be represented as a collection of binary symmetric channels (BSC's). The binary input is given to one of these BSC's at random such that the $i$-th BSC is chosen with probability $p_i$. The output of this BSC together with its cross over probability $x_i$ is considered as the output of the channel. Therefore, a discrete BMS channel $W$ can be completely described by a random variable $\chi \in [0,1/2]$. The \emph{pdf} of $\chi$ will be of the form:
\begin{align}
P_\chi(x) = \sum_{i=1}^m p_i \delta(x - x_i)
\end{align}
such that $\sum_{i=1}^m p_i = 1$ and $0 \leq x_i \leq 1/2$. Note that $Z(W)$ and $1-I(W)$ are expectations of the functions $f(x)=2\sqrt{x(1-x)}$ and $g(x)=-x\log(x)-(1-x) \log(1-x)$ over the distribution $P_{\chi}$, respectively.  
 
Therefore, in the quantization problem we want to replace the mass distribution $P_\chi$ with another mass distribution $P_{\tilde{\chi}}$ such that the number of output symbols of $\tilde{\chi}$ is at most $k$, and the channel $\tilde{W}$ is polar degraded with respect to $W$. We know that the following two operations imply polar degradation: 
\begin{itemize}
\item Stochastically degrading the channel.
\item Replacing the channel with a BEC channel with the same Bhattacharyya parameter.
\end{itemize}
Furthermore, note that the \emph{stochastic dominance} of random variable $\tilde{\chi}$ with respect to $\chi$ implies $\tilde{W}$ is stochastically degraded with respect to $W$. (But the reverse is not true.)  

In the following, we propose different algorithms based on different methods of polar degradation of the channel. The first algorithm is a naive algorithm called the mass transportation algorithm based on the stochastic dominance of the random variable $\tilde{\chi}$, and the second one which outperforms the first is called greedy  mass merging algorithm. For both of the algorithms the quantized channel is stochastically degraded with respect to the original one. 
\subsection{Greedy Mass Transportation Algorithm}
 In the most general form of this algorithm we basically look at the problem as a \emph{mass transport} problem. In fact, we have non-negative masses $p_i$ at locations $x_i, i=1,\cdots,m, x_1 < \cdots < x_m$.
What is required is to move the masses, by only moves to the right, to concentrate them on $k < m$ locations, and try to minimize $\sum_i p_i d_i$ where $d_i=x_{i+1} - x_i$ is the amount $i^{th}$ mass has moved.  
Later, we will show that this method is not optimal but useful in the theoretical analysis of the algorithms that follow.

\begin{algorithm}
\caption{Mass Transportation Algorithm}
\label{mass}
\begin{algorithmic}[1]
\STATE Start from the list $(p_1,x_1), \cdots , (p_m,x_m)$.
\STATE Repeat $m-k$ times
\STATE Find $j = \text{argmin} \{ p_i d_i: \, i \neq m \} $
\STATE Add $p_j$ to $p_{j+1}$ (i.e. move $p_j$ to $x_{j+1}$)
\STATE Delete $(p_j, x_j)$ from the list.

\end{algorithmic}
\end{algorithm}

Note that Algorithm \ref{mass} is based on the stochastic dominance of random variable $\tilde{\chi}$ with respect to $\chi$. Furthermore, in general, we can let $d_i = f(x_{i+1}) - f(x_i)$, for an arbitrary increasing function $f$. 

\subsection{Mass Merging Algorithm}

The second algorithm merges the masses. Two masses $p_1$ and $p_2$ at positions $x_1$ and $x_2$ would be merged into one mass $p_1 + p_2$ at position $\bar{x}_1 = \frac{p_1}{p_1 + p_2} x_1 + \frac{p_2}{p_1 + p_2} x_2$. This algorithm is based on the stochastic degradation of the channel, but the random variable $\chi$ is not stochastically dominated by $\tilde{\chi}$. The greedy algorithm for the merging of the masses would be the following:

\begin{algorithm}[h!]
\caption{Merging Masses Algorithm}
\label{merge}
\begin{algorithmic}[1]
\STATE Start from the list $(p_1,x_1), \cdots , (p_m,x_m)$.
\STATE Repeat $m-k$ times
\STATE Find $j = \text{argmin} \{ p_i ( f(\bar{x}_i) - f(x_i) ) - p_{i+1} ( f(x_{i+1}) - f(\bar{x}_i) ): ~ i \neq m \} \quad \bar{x}_i = \frac{p_i}{p_i + p_{i+1}} x_i + \frac{p_{i+1}}{p_i + p_{i+1}} x_{i+1}$
\STATE Replace the two masses $(p_j, x_j)$ and $(p_{j+1}, x_{j+1})$ with a single mass $(p_j + p_{j+1} , \bar{x}_j)$.

\end{algorithmic}
\end{algorithm}

Note that in practice, the function $f$ can be any increasing concave function, for example, the entropy function or the Bhattacharyya function. In fact, since the algorithm is greedy and suboptimal, it is hard to investigate explicitly how changing the function $f$ will affect the total error of the algorithm in the end (i.e., how far $\tilde{W}$ is from $W$). 

\section{Bounds on the Approximation Loss}\label{sec:error1}
In this section, we provide some bounds on the maximum approximation loss we have in the algorithms. We define the ``approximation loss'' to be the difference between the expectation of the function $f$ under the true distribution $P_\chi$ and the approximated distribution $P_{\tilde{\chi}}$. Note that the kind of error that is analyzed in this section is different from what was defined in Section \ref{sec:prob}. The connection of the approximation loss with the quantization error is made clear in Theorem \ref{thm:1}. For convenience, we will simply stick to the word ``error'' instead of ``approximation loss'' from now on.  

We first find an upper bound on the error made in Algorithms \ref{mass} and \ref{merge} and then use it to provide bounds on the error made while performing operations \eqref{1} and \eqref{2}. 

\begin{lem}\label{lem:error}
The maximum error made by Algorithms \ref{mass} and \ref{merge} is upper bounded by $\mathcal{O}(\frac{1}{k})$.
\end{lem}

\begin{proof}
First, we derive an upper bound on the error of Algorithms \ref{mass} and \ref{merge} in each iteration, and therefore a bound on the error of the whole process. Let us consider Algorithm \ref{mass}. The problem can be reduced to the following optimization problem:
\begin{align}
e = \max_{p_i,x_i} \min_{i} (p_i d_i)
\end{align}
such that 
\begin{align} 
\sum_{i} p_i = 1, ~ \sum_{i} d_i & \leq 1,
\end{align}
where $d_i = f(x_{i+1}) - f({x_i})$, and $f(\frac 12)-f(0)=1$ is assumed w.l.o.g. We prove the lemma by Cauchy-Schwarz inequality.
\begin{align}
\min_{i} p_i d_i = \left( \sqrt{\min_{i} p_i d_i} \right )^2 
= \left(\min_{i} \sqrt{p_i d_i} \right) ^2
\end{align}
Now by applying Cauchy-Schwarz we have
\begin{align}
\sum_{i=1}^{m} \sqrt{ p_i d_i} \leq \left(\sum_{i=1}^{m} p_i \right)^{1/2} \left(\sum_{i=1}^{m} d_i \right)^{1/2} \leq 1 
\end{align}
Since the sum of $m$ terms $\sqrt{ p_i d_i}$ is less than $1$, the minimum of the terms will certainly be less than $\frac{1}{m}$. Therefore,  
\begin{align}\label{10}
e = \left( \min \sqrt{p_i d_i}\right) ^2 \leq \frac{1}{m^2}.
\end{align}
For Algorithm \ref{merge}, achieving the same bound as Algorithm \ref{mass} is trivial. Denote $e^{(1)}$ the error made in Algorithm \ref{mass} and $e^{(2)}$ the error made in Algorithm \ref{merge}. Then,

\begin{align}
e_i^{(2)} & = p_i \left( f(\bar{x}_i) - f(x_i) \right) - p_{i+1} \left( f(x_{i+1}) - f(\bar{x}_i) \right) \\
& \leq p_i \left( f(\bar{x}_i) - f(x_i) \right) \\ \label{outperform}
& \leq p_i \left( f(x_{i+1}) - f(x_i) \right) = e_i^{(1)}.
\end{align}

Consequently, the error generated by running the whole algorithm can be upper bounded by
$\sum_{i = k+1}^m \frac{1}{i^2} $ which is $\mathcal{O}(\frac{1}{k})$.

\end{proof} 
  
What is stated in Lemma \ref{lem:error} is a loose upper bound on the error of Algorithm \ref{merge}. To achieve better bounds, we upper bound the error made in each iteration of the Algorithm \ref{merge} as the following:
\begin{align}
e_i &= p_i \left( f(\bar{x}_i) - f(x_i) \right) - p_{i+1} \left( f(x_{i+1}) - f(\bar{x}_i) \right)   \\\label{concave}
& \leq p_i \frac{p_{i+1}}{p_i + p_{i+1}} \Delta x_i f'(x_i) - p_{i+1} \frac{p_i}{p_i + p_{i+1}} \Delta x_i f'(x_{i+1}) \\
&= \frac{p_i p_{i+1}}{p_i + p_{i+1}} \Delta x_i \left(f'(x_i) - f'(x_{i+1}) \right) \\ \label{MVT} 
& \leq \frac{p_i + p_{i+1}}{4} \Delta x_i^2 |f''(c_i)|,
\end{align}
where $\Delta x_i = x_{i+1} - x_{i}$ and \eqref{concave} is due to concavity of function $f$. Furthermore, \eqref{MVT} is by the mean value theorem, where $x_i \leq c_i \leq x_{i+1}$.

If $|f''(x)|$ is bounded for $x \in (0,1)$, then we can prove that $\min_{i} e_i \sim \mathcal{O}(\frac{1}{m^3})$ similarly to Lemma \ref{lem:error}. Therefore the error of the whole algorithm would be $\mathcal{O}(\frac{1}{k^2})$. Unfortunately, this is not the case for either of entropy function or Bhattacharyya function. However, we can still achieve a better upper bound for the error of Algorithm \ref{merge}.     

\begin{lem}\label{lem:error2}
The maximum error made by Algorithm \ref{merge} for the entropy function $h(x)$ can be upper bounded by the order of $\mathcal{O}(\frac{\log(k)}{k^{1.5}})$.
\end{lem}

\begin{proof}
See Appendix. 
\end{proof}

\vskip 5pt

We can see that the error is improved by a factor of $\frac{\log{k}}{\sqrt{k}}$ in comparison with Algorithm \ref{mass}.

Now we use the result of Lemma \ref{lem:error} to provide bounds on the total error made in estimating the mutual information of a channel after $n$ levels of operations \eqref{1} and \eqref{2}.

\begin{thm}\label{thm:1}
Assume $W$ is a BMS channel and using Algorithm~\ref{mass} or \ref{merge} we quantize the channel $W$ to a channel $\tilde{W}$. Taking $k=n^2$ is sufficient to give an approximation error that decays to zero. 
\end{thm}

\begin{proof}
First notice that for any two BMS channels $W$ and $V$, doing the polarization operations \eqref{1} and \eqref{2}, the following is true: 
\begin{equation}\label{eq:I}
(I(W^-) - I(V^-)) + (I(W^+) - I(V^+)) = 2 (I(W) - I(V))
\end{equation}
Replacing $V$ with $\tilde{W}$ in \eqref{eq:I} and using the result of Lemma \ref{lem:error}, we conclude that after $n$ levels of polarization the sum of the errors in approximating the mutual information of the $2^n$ channels is upper-bounded by $\mathcal{O}(\frac{n2^n}{k})$. In particular, taking $k=n^2$, one can say that the ``average'' approximation error of the $2^n$ channels at level $n$ is upper-bounded by $\mathcal{O}(\frac{1}{n})$. Therefore, at least a fraction $1-\frac{1}{\sqrt{n}}$ of the channels are distorted by at most $\frac{1}{\sqrt{n}}$ i.e., except for a negligible fraction of the channels the error in approximating the mutual information decays to zero. 

\end{proof}

As a result, since the overall complexity of the encoder construction is $\mathcal{O}(k^2N)$, this leads to ``almost linear'' algorithms for encoder construction with arbitrary accuracy in identifying good channels.

\section{Exchange of Limits} \label{sec:ex}
In this section, we show that there are admissible schemes such that as $k \to \infty$, the limit in \eqref{ro} approaches $I(W)$ for any BMS channel $W$. We use the definition stated in \eqref{ad2} for the admissibility of the quantization procedure.
\begin{thm} \label{thm2}
Given a BMS channel $W$ and for large enough $k$, there exist admissible quantization schemes $Q_k$ such that $\rho(Q_k, W)$ is arbitrarily close to $I(W)$.
\end{thm}
\begin{proof}
Consider the following algorithm:  The algorithm starts with a quantized version of $W$ and it does the normal channel 
splitting transformation followed by quantization according to Algorithm~\ref{mass} or \ref{merge}, but
once a sub-channel is sufficiently good, in the sense that its Bhattacharyya parameter is less than an appropriately chosen parameter $\delta$, the algorithm replaces the sub-channel with a binary erasure channel which is  degraded (polar degradation) with respect to it (As the operations \eqref{1} and \eqref{2} over an erasure channel also yields and erasure channel, no further quantization is need for the children of this sub-channel). 

Since the ratio of the total good indices of BEC($Z(P)$) is $1-Z(P)$, then the total error that we make by replacing $P$ with BEC($Z(P)$) is at most $Z(P)$
which in the above algorithm is less that the parameter $\delta$.   

Now, for a fixed level $n$, according to Theorem \ref{thm:1} if we make $k$ large enough, the ratio of the quantized sub-channels that their Bhattacharyya value is less that $\delta$ approaches to its original value (with no quantization), and for these sub-channels as explained above the total error made with the algorithm 
is $\delta$. Now from the polarization theorem and by sending $\delta$ to zero we deduce that as $k \to \infty$ the number of good indices approaches the capacity of the original channel.  
\end{proof}


\section{Simulation Results}

In order to evaluate the performance of our quantization algorithm, similarly to \cite{TV11}, we compare the performance of the degraded quantized channel with the performance of an upgraded quantized channel. An algorithm similar to Algorithm \ref{merge} for upgrading a channel is the following. Consider three neighboring masses in positions $(x_{i-1},x_i,x_{i+1})$ with probabilities $(p_{i-1},p_i,p_{i+1})$. Let $t = \frac{x_i - x_{i-1}}{x_{i+1}-x_{i-1}}$. Then, we split the middle mass at $x_i$ to the other two masses such that the final probabilities will be $(p_{i-1}+(1-t)p_i,p_{i+1}+t p_i)$ at positions $(x_{i-1},x_{i+1})$. The greedy channel upgrading procedure is described in Algorithm~\ref{split}.  

\begin{algorithm}
\caption{Splitting Masses Algorithm}
\label{split}
\begin{algorithmic}[1]
\STATE Start from the list $(p_1,x_1), \cdots , (p_m,x_m)$.
\STATE Repeat $m-k$ times
\STATE Find $ j = \text{argmin} \{ p_i ( f(x_i) - tf(x_{i+1}) - (1-t)f(x_{i-1}) ): i \neq 1,m \} $
\STATE Add $(1-t)p_j$ to $p_{j-1}$ and $tp_j$ to $p_{j+1}$.
\STATE Delete $(p_j, x_j)$ from the list.

\end{algorithmic}
\end{algorithm} 

The same upper bounds on the error of this algorithm can be provided similarly to Section \ref{sec:error1} with a little bit of modification. 

In the simulations, we measure the maximum achievable rate while keeping the probability of error less than $10^{-3}$ by finding maximum possible number of channels with the smallest Bhattacharyya parameters such that the sum of their Bhattacharyya parameters is upper bounded by $10^{-3}$. The channel is a binary symmetric channel with capacity $0.5$. Using Algorithms \ref{merge} and \ref{split} for degrading and upgrading the channels with the Bhattacharyya function $f(x) = 2 \sqrt{x(1-x)}$, we obtain the following results: 

\begin{table}[h!]
\centering

\begin{tabular}{c | c c c c c c}

$k$ & $2$ & $4$ & $8$&  $16$ &   $32$  &  $64$  \\
  
 \hline
degrade &  $0.2895$ &   $0.3667$ & $0.3774 $&  $0.3795$ &  $0.3799$  &$0.3800$  \\ 
upgrade  & $0.4590$ &  $0.3943$ &$0.3836$  &$0.3808$   &$0.3802$ & $0.3801$ 
\end{tabular}

\caption{{\small Achievable rate with error probability at most $10^{-3}$ vs. maximum number of output symbols $k$ for block-length $N=2^{15}$ }}

\vskip 0.15cm
\label{table1}
\end{table} 

It is worth restating that the algorithm runs in complexity $\mathcal{O}(k^2 N)$. Table~\ref{table1} shows the achievable rates for Algorithms \ref{merge} and \ref{split} when the block-length is fixed to $N=2^{15}$ and $k$ changes in the range of $2$ to $64$.     

It can be seen from Table~\ref{table1} that the difference of achievable rates within the upgraded and degraded version of the scheme is as small as $10^{-4}$ for $k=64$. We expect that for a fixed $k$, as the block-length increases the difference will also increase (see Table~\ref{table2}).
\begin{table}[h!]
\centering

\begin{tabular}{c | c c c c c c}

$n$ & $5$ & $8$ & $11$&  $14$ &   $17$  &  $20$\\
  
 \hline
degrade &  $0.1250$ &   $0.2109$ & $0.2969 $&  $0.3620$ &  $0.4085$ & $0.4403$ \\ 
upgrade  & $0.1250$ &  $0.2109$ &$0.2974$  &$0.3633$   &$0.4102$ & $0.4423$
\end{tabular}

\caption{{\small Achievable rate with error probability at most $10^{-3}$ vs.  block-length $N=2^n$ for $k=16$ }}

\vskip 0.15cm
\label{table2}
\end{table} 

 However, in our scheme this difference will remain small even as $N$ grows arbitrarily large as predicted by Theorem \ref{thm2}. (see Table~\ref{table3}).

\begin{table}[h!]
\centering

\begin{tabular}{c | c c c c c}

$n$ & $21$ & $22$ & $23$&  $24$ &   $25$\\
  
 \hline

degrade &  $0.4484$ &   $0.4555$ & $0.4616 $&  $0.4669$ &  $0.4715$\\ 
upgrade  & $0.4504$ &  $0.4575$ &$0.4636$  &$0.4689$   &$0.4735$
\end{tabular}

\caption{{\small Achievable rate with error probability at most $10^{-3}$ vs.  block-length $N=2^n$ for $k=16$ }}

\vskip 0.15cm
\label{table3}
\end{table} 

We see that the difference between the rate achievable in the degraded channel and upgraded channel gets constant $2 \times 10^{-3}$ even after $25$ levels of polarizations for $k=16$.  

\appendix

\subsection{Proof of Lemma \ref{merge}}
\begin{proof}
Let us first find an upper bound for the second derivative of the entropy function. Suppose that $h(x) = -x\log(x) - (1-x)\log(1-x)$. Then, for $0 \leq x \leq \frac{1}{2}$, we have
\begin{align}\label{derivative}
|h''(x)| = \frac{1}{x(1-x)\ln(2)} \leq \frac{2}{x \ln(2)}.
\end{align}
Using \eqref{derivative} the minimum error can further be upper bounded by
\begin{align}
\min_{i} e_{i} \leq \min_{i} (p_i + p_{i+1}) \Delta x_i^2 \frac{1}{x_i \ln(4)}.
\end{align}
Now suppose that we have $l$ mass points with $x_i \leq \frac{1}{\sqrt{m}}$ and $m-l$ mass points with $x_i \geq \frac{1}{\sqrt{m}}$. For the first $l$ mass points we use the upper bound obtained for Algorithm \ref{mass}. Hence, for $1 \leq i \leq l$ we have
\begin{align} \label{l1}
\min_{i} e_{i} & \leq \min_{i} p_i \Delta h(x_i) \\ \label{l2}
& \sim \mathcal{O} \left( \frac{\log(m)}{l^2 \sqrt{m}} \right),
\end{align}
where \eqref{l1} is due to \eqref{outperform} and \eqref{l2} can be derived again by applying Cauchy-Schwarz inequality. Note that this time 
\begin{align}
\sum_{i=1}^{l} \Delta h(x_i) \leq h(\frac{1}{\sqrt{m}}) \sim \mathcal{O} \left( \frac{\log(m)}{\sqrt{m}} \right).
\end{align}

For the $m-l$ mass points one can write
\begin{align}
\min_{i} e_i &\leq \min_{i} (p_i + p_{i+1}) \Delta x_i^2 \frac{1}{x_i \ln(4)} \\
& \leq  \min_{i} (p_i + p_{i+1}) \Delta x_i^2 \frac{\sqrt{m}}{\ln(4)} \\ \label{holder}
& \sim \mathcal{O} \left(\frac{\sqrt{m}}{(m-l)^3}\right),
\end{align}
where \eqref{holder} is due to H\"{o}lder's inequality as follows:

Let $q_i = p_i + p_{i+1}$. Therefore, $\sum_{i}(p_i + p_{i+1}) \leq 2$ and $\sum_{i} \Delta x_i \leq 1/2$. 
\begin{equation}
\min_{i} q_i \Delta x_i^2 = \left( \left(\min_{i} q_i \Delta x_i^2 \right)^{1/3} \right )^3 
= \left(\min_{i} \left( q_i \Delta x_i^2 \right)^{1/3} \right) ^3
\end{equation}
Now by applying H\"{o}lder's inequality we have
\begin{align}
\sum_{i} \left( q_i \Delta x_i^2 \right)^{1/3} \leq \left(\sum_{i} q_i \right)^{1/3} \left(\sum_{i} \Delta x_i \right)^{2/3} \leq 1 
\end{align}
Therefore,  
\begin{align}
\min_i e_{i} \leq \sqrt{m} \left( \min_{i} ( q_i \Delta x_i^2)^{1/3} \right) ^3 \sim \mathcal{O} \left( \frac{\sqrt{m}}{(m-l)^3} \right) .
\end{align}

Overall, the error made in the first step of the algorithm would be 
\begin{align}
\min_i e_i & \sim \min \left \{ \mathcal{O} \left( \frac{\log(m)}{l^2 \sqrt{m}} \right),\mathcal{O} \left(\frac{\sqrt{m}}{(m-l)^3} \right) \right \} \\
& \sim \mathcal{O} \left(\frac{\log(m)}{m^{2.5}}\right).
\end{align}
Thus, the error generated by running the whole algorithm can be upper bounded by $\sum_{i = k+1}^m \frac{\log(i)}{i^{2.5}}
  \sim O\left(\frac{\log(k)}{k^{1.5}}\right)$.

\end{proof}

\section*{Acknowledgments}

{\small EPFL authors are grateful to R\"{u}diger Urbanke for helpful discussions. 
This work was supported in part by grant number $200021$-$125347$ of the Swiss National Science Foundation.}


\begin{thebibliography}{X}
\bibitem{ari09}
E.~Ar{\i}kan, ``Channel Polarization: A Method for Constructing Capacity-Achieving Codes for Symmetric Binary-Input Memoryless Channels,''
\emph{{IEEE} Trans. Inf. Theory}, vol.~55, no.~7, pp. 3051--3073, Jul. 2009.
\bibitem{SK09}
S.~B.~Korada, ``Polar Codes for Channel and Source Coding,'' Ph.D. dissertation, EPFL, Lausanne, Switzerland, Jul. 2009. 
\bibitem{TV11}
I.~Tal and A.~Vardy, ``How to Construct Polar Codes,'' [Online]. Available:
  \url{http://arxiv.org/pdf/1105.6164}.
\bibitem{Compound}
S.~H.~Hassani, S.~B.~Korada, and R.~Urbanke, ``The Compound Capacity of Polar Codes,'' \emph{Proceedings of Allerton Conference on Communication, Control and Computing}, Allerton, Sep. 2009.
\bibitem{Tanaka}
R.~Mori and T.~Tanaka, ``Performance and Construction of Polar Codes on Symmetric Binary-Input Memoryless Channels,'' \emph{Proceedings of ISIT}, Seoul, South Korea, Jul. 2009, pp. 1496--1500.  
\bibitem{Urbanke}
T.~Richardson and R.~Urbanke, ``Modern Coding Theory,'' Cambridge University Press, 2008.  

\end{thebibliography}
\end{document}